\documentclass[journal]{IEEEtran}
\ifCLASSINFOpdf
\else
\fi
\usepackage[dvips]{color}
\usepackage{graphicx}

\usepackage[ruled,vlined,boxed]{algorithm2e}
\usepackage{multirow,epsfig,amsmath,amsfonts,cite,times}
\usepackage{graphicx}

\usepackage{subfigure}
\usepackage{algorithmic}
\usepackage{multirow}

\newtheorem{theorem}{Theorem}[section]
\newtheorem{lemma}[theorem]{Lemma}

\makeatletter

\begin{document}

\title{An Efficient Topology-Based Algorithm for Transient Analysis of Power Grid}

\author{
Lan Yang, Jingbin Wang, Lorenzo Azevedo, and Jim Jing-Yan Wang
\thanks{
Lan Yang is with the School of Computer Science and Information Engineering, Chongqing Technology and
Business University, Chongqing 400067, China.}
\thanks{
Jingbin Wang is with the 
Tianjin Key Laboratory of Cognitive Computing and Application, Tianjin University,
Tianjin 300072, China.}
\thanks{Lorenzo Azevedo is with the 
Computer Science Department, University of S{\"a}o Paulo, S?o Paulo - SP, 05508-070,
Brazil}
\thanks{
Jim Jing-Yan Wang is with the Computer, Electrical and Mathematical Sciences and Engineering Division, King Abdullah University of Science and Technology (KAUST), Thuwal 23955-6900, Saudi Arabia.}
}

\markboth{Wang: An efficient topological-based power grid analysis algorithm using nodal analysis}%
{Wang: An efficient topological-based power grid analysis algorithm using nodal analysis}
\maketitle

\begin{abstract}
In the design flow of integrated circuits, chip-level verification is an important step that sanity checks the performance is as expected. Power grid verification is one of the most expensive and time-consuming steps of chip-level verification, due to its extremely large size. Efficient power grid analysis technology is highly demanded as it saves computing resources and enables faster iteration. In this paper, a topology-base power grid transient analysis algorithm is proposed. Nodal analysis is adopted to analyze the topology which is mathematically equivalent to iteratively solving a positive semi-definite linear equation. The convergence of the method is proved.
\end{abstract}

\begin{IEEEkeywords}
Power grid, topology, nodal analysis, positive semi-definite
\end{IEEEkeywords}

\IEEEpeerreviewmaketitle

\section{Introduction}

\IEEEPARstart{P}ower grid analysis is an indispensable step in modern chip simulation and verification. Due to the $IR$ and $Ldi/dt$ voltage drops, the actual voltages applied to the logic gates and other circuit elements are smaller than wished, which not only increases the delay of the logic gates but may also result in logic errors. As the ever decreasing supply voltages and threshold voltages make the problem worse, power grid analysis becomes an indispensable step in the simulation flow.

To capture the worst-case voltage drop and place decoupling capacitors (if necessary), DC analysis and/or transient analysis of the power grid should be performed. Frequency-domain analysis is not appropriate as the vital voltage drop in the time domain may be lost in the frequency-domain analysis. Both DC analysis and transient analysis (using forward/backwardcan be finally reduced to a linear solving problem

Traditional power grid analysis only considers the $IR$ voltage drop. As the operation frequency of chips is increasing rapidly, the $Ldi/dt$ voltage drop should not be neglected. Thus, the power grid model should consist inductors besides resistors and capacitors. The existence of inductors changes the state-space model from order-1 to order-2, hence introduces new difficulties to the transient analysis and its topology-based algorithms. It was proved in \cite{chen2001efficient} that the system matrix of the RLC model could be still positive definite, which had been thought to be impossible before the paper. However, some of its details are inaccurate. The right formulation will be proposed in Section \ref{sec_RLCcon}, which is the basis to prove the convergence of the topology-based iterative algorithm in Section \ref{sec_RLCalg}.

In this paper, a new topology-based algorithm is proposed. The Nodal analysis is adopted to analyze the netlist and generate positive semi-definite system model. Then a purely topology-base iterative method is introduced to solve the power grid analysis without constructing the system matrices. The convergence is proved and the complexity is analyzed.

\section{Background}
\subsection{Problem formulation}
Power grids are usually modeled as RLC circuits, or simpler RC circuits. The current sources attached to the nodes of the power grid represent the currents drawn from the logic gates or other devices. Hence prior to the simulation of the power grid, the current pattern of the gates and devices should be obtained and modeled as piecewise linear ideal current sources. The capacitors in parallel with the currents sources can be either parasitic capacitors or decoupling capacitors. The inductors most possibly appear at the vias connecting different layers and the metal layer connecting the power grid to external power supply. In many models the inductors can be neglected.

First consider the simpler RC power grid model, which can be depicted as a state-space model \cite{kouroussis2003static, wang2010mfti,lim2010ultra,lei2010vector}
\begin{equation}\label{eqn_RC}
C\frac{d}{dt}v(t)+Gv(t)=-i(t)+G_0V_{dd},
\end{equation}
where $v(t)$ is the vector of nodal voltages, $i(t)$ is the vector of current sources, $V_{dd}$ is vector of dimension $n$ with all its values equal to $V_{dd}$. $C$ is a diagonal matrix whose $(i,i)$ element is the capacitance between node $i$ and ground. $G$ is the conductance matrix with its diagonal $(i,i)$ element being the total conductance connected to node $i$ and off-diagonal $(i,j)$ element being the opposite number of the conductance between node $i$ and node $j$ (it is zero if node $i$ and node $j$ are unconnected). $G_0$ is a diagonal matrix with its $(i,i)$ element being the conductance between node $i$ and source (if node $i$ is connected to a source) or $0$ (if node $i$ is not connected to any source). To simplify the notation, $b(t)$ is used to represent $-i(t)+G_0V_{dd}$. Using backward Euler method, (\ref{eqn_RC}) can be approximated as
\begin{equation}\label{eqn_RCtran}
(\frac{C}{h}+G)x(t+h)=\frac{C}{h}x(t)+b(t+h),
\end{equation}
where $h$ represents the time step.
Thus given the initial condition $x(0)$, we can calculate all the $x(t)$ after $0$. Note that $\frac{C}{h}+G$ is always invertible if the system is stable. The reason is that stability implies that all the eigenvalues of the matrix pencil $(G,C)$ distribute in the left half complex plane and the positive real value $\frac{1}{h}$ cannot be an eigenvalue of $(G,C)$.

If we use system matrix $A$ to denote $\frac{C}{h}+G$ and $b$ to denote $\frac{C}{h}x(t)+b(t+h)$, (\ref{eqn_RCtran}) can be reduced to
\begin{equation}\label{eqn_RCsys}
Ax=b.
\end{equation}
In DC analysis, $A=G$ and $b$.

\subsection{SOR-like iterative method}
Although straightforward in theory, Eq. (\ref{eqn_RCtran}) are difficult to solve due to the extremely large size. In a standard power grid analysis problem, the number of nodes $n$ can be hundreds of thousands or even millions, thus the dimension of the matrices can be up to millions. This causes problems owing to the limited speed and memory of computers. Direct methods based on LU decomposition requires $O(n^3)$ time for the decomposition of the system matrix ($G$ or $A$) and $O(n^2)$ time for forward and backward substitution. The memory required is $O(n^2)$ even if $G$ or $A$ is sparse as the resulting matrices $L, U$ may become dense. Therefore direct methods are not applicable for extremely large problems.

Alternative methods include traditional iterative methods, the widely used PCG (preconditioned conjugate gradient) method and the Monte-Carlo-like random walk method, etc. PCG method, although converges fast, may be not applicable to extremely large problems as preconditioned system matrix may become dense and thus requires prohibiting $O(n^2)$ memory. Random walk algorithm is most suitable in the case that we only want to calculate the voltages at some specific nodes, but it may be inefficient for the full-circuit analysis, which is indispensable to find the worst-case voltage drop.

An iterative method was introduced in \cite{zhong2005fast}, which can be regarded as an efficient implementation of the traditional Gauss-Seidel iteration method and SOR (Successive Over-Relaxation) method. The convergence of this method in the DC analysis has been proven in \cite{zhong2005fast}. The advantage of this method is that it is topology-based and no matrix construction is needed. Hence it requires much less memory than PCG method. It is also shown that the method is faster than random walk in \cite{zhong2005fast}.

In the next sections, we will extend this SOR-like method to the transient analysis of RC and RLC power grid models and prove their convergence. In the RLC model analysis part, we will first show that the proof of positive definiteness in \cite{chen2001efficient} is inaccurate and then give the right version.

\section{Nodal Analysis of RLC circuits}
In the analysis of RLC circuits the most commonly used method is MNA (modified nodal analysis). It can be guaranteed that these system models are passive \cite{wang2010peds, wang2012passivity} By introducing extra variables for currents flowing through voltage sources and inductors, it can generate compact state-space models efficiently. However, when we perform transient analysis on such MNA models, the system matrix is not positive definite. As a result, many iterative algorithms (such as preconditioned conjugate gradient) cannot be applied since their convergence requires positive definiteness of the system matrix.

In \cite{chen2001efficient}, a new NA (nodal analysis) method was proposed to perform transient analysis on RLC circuits. It was proved in \cite{chen2001efficient} that the resulting system matrix is guaranteed to be positive definite. Although the idea in the paper is novel and inspiring, its formulations is inaccurate. The key equation (equation (8) in \cite{chen2001efficient}) is incorrect due to a wrong sign in equation (6). From another perspective, in equation (8), the current state of the nodal voltages only depends on the state of the previous step, which is impossible for a second-order system. In a second-order system like RLC circuits, the required initial conditions involve both $x(0)$ and $x'(0)$. In the discretized form, the state $x(t+\Delta t)$ should depend on both $x(t)$ and $x(t-\Delta t)$, which is not the case in \cite{chen2001efficient}. Besides, the analysis in \cite{chen2001efficient} did not consider voltage sources. Although Norton equivalent was used to convert voltages sources to current sources, we still wish to involve voltages sources from the beginning of the deduction, as Norton equivalent does not work for ideal voltage sources which we may want to include in the power grid model. We will give the correct deduction considering ideal voltage sources in the next subsection.

\subsection{System equation}

To simplify the deduction, assume that all the negative terminals of the voltage sources are connected to ground, which is just the case in the power grid analysis (it can be extended to ground network analysis straightforwardly). Assume that there are $n$ nodes that are neither ground nor positive terminals of voltage sources in the power grid model (called trivial nodes), numbered from $1$ to $n$, together with $p$ nodes being the positive terminals of voltage sources (called source nodes), numbered from $1$ to $p$. Besides, assume there are $m$ branches not through voltage sources, numbered from $1$ to $m$. Define the $m$ by $n$ modified incidence matrix $A$ as
\begin{equation}\label{eqn_incid}
A(i,j)=\left\{ {\begin{array}{*{20}{c}}
   { + 1} ,~~~\mbox{if trival node $j$ is the source of branch $i$;} \\
   { - 1} ,~~~\mbox{if trival node $j$ is the sink of branch $i$;} \\
   0 ,~~~ \mbox{otherwise.} \\
\end{array}} \right.
\end{equation}

Similarly, define the $m$ by $p$ source incidence matrix $A_s$ as
\begin{equation}\label{eqn_sincid}
A_s(i,j)=\left\{ {\begin{array}{*{20}{c}}
   { + 1} ,~~~\mbox{if source node $j$ is the source of branch $i$;} \\
   { - 1} ,~~~\mbox{if source node $j$ is the sink of branch $i$;} \\
   0 ,~~~ \mbox{otherwise.} \\
\end{array}} \right.
\end{equation}

The advantage of splitting the traditional incidence matrix to the modified incidence matrix and the source incidence matrix is that no extra current variables through the voltage sources needs to be introduced, which facilities the nodal analysis and is the basis of the positive definiteness to be proved next. By grouping the branches with resistors, inductors, capacitors and current sources together we obtain
\begin{equation}\label{eqn_group}
A = \left[ {\begin{array}{*{20}{c}}
   {{A_g}}  \\
   {{A_c}}  \\
   {{A_l}}  \\
   {{A_i}}  \\
\end{array}} \right], {A_s} = \left[ {\begin{array}{*{20}{c}}
   {{A_{gs}}}  \\
   {{A_{cs}}}  \\
   {{A_{ls}}}  \\
   {{A_{is}}}  \\
\end{array}} \right], {v_b} = \left[ {\begin{array}{*{20}{c}}
   {{v_g}}  \\
   {{v_c}}  \\
   {{v_l}}  \\
   {{v_i}}  \\
\end{array}} \right], {i_b} = \left[ {\begin{array}{*{20}{c}}
   {{i_g}}  \\
   {{i_c}}  \\
   {{i_l}}  \\
   {{i_i}}  \\
\end{array}} \right],
\end{equation}
where $i_b$ is the vector of branch currents, $v_b$ is the vector of branch voltages, $g$, $c$, $l$, $s$ represent resistors, capacitors, inductors and current sources respectively.

Applying Kirchoff's current and voltage laws, we obtain
\begin{equation}\label{eqn_kir}
\begin{gathered}
A^Ti_b=0,\\
Av_n=v_b-A_sv_d,
\end{gathered}
\end{equation}
where $v_n$ is the vector of nodal voltages (other than the nodes being the terminals of voltage sources) and $v_d$ is the (constant) vector of voltage sources.

Besides, the branch currents and branch voltages follow the following branch equations (grouped by branch elements)
\begin{equation}\label{eqn_branch}
\begin{gathered}
i_g=\mathcal{G}v_g,\\
i_c=\mathcal{C}\frac{dv_c}{dt},\\
v_l=\mathcal{L}\frac{di_l}{dt},\\
i_i=I_s,\\
\end{gathered}
\end{equation}
where $\mathcal{G}$, $\mathcal{C}$, $\mathcal{L}$ are diagonal matrices with its diagonal elements being the value of corresponding branch resistance, capacitance and inductance, $I_s$ is the vector of branch current sources.

Combine (\ref{eqn_kir}) and (\ref{eqn_branch}), we obtain $2m+n$ equations of $2m+n$ variables. Substitute the second equation of (\ref{eqn_kir}) to (\ref{eqn_branch}) and then substitute the new (\ref{eqn_branch}) to the first equation of (\ref{eqn_kir}), we have
\begin{equation}\label{eqn_simp}
\begin{gathered}
C\frac{dv_n}{dt}+Gv_n+A_l^Ti_l+A_i^TI_s-G_sv_d=0,\\
L\frac{di_l}{dt}-A_lv_n+A_{ls}v_d=0.
\end{gathered}
\end{equation}
Here $G=A_g^T\mathcal{G}A_g$, $C=A_c^T\mathcal{C}A_c$, $G_s=A_g^T\mathcal{G}A_{gs}$. Use backward Euler law to discretize (\ref{eqn_simp}),

\begin{equation}\label{eqn_discrete}
\begin{gathered}
C\frac{v_n(t+h)-v_n(t)}{h}+Gv_n(t+h)+A_l^Ti_l(t+h)\\
~~~~~~~~~~~~~~~~~~~~~~~~~+A_i^TI_s(t+h)-G_sv_d=0,\\
\mathcal{L}\frac{i_l(t+h)-i_l(t)}{h}-A_lv_n(t+h)+A_{ls}v_d=0.\\
\end{gathered}
\end{equation}
Here $h$ is the selected time step length. Rewrite the first equation of (\ref{eqn_discrete}) at time $t$ and subtract it from the original equation, then substitute the second equation of (\ref{eqn_discrete}) to it we obtain

\begin{equation}\label{eqn_systemeqn}
\begin{gathered}
\left( \frac{C}{h}+G+hL \right)v_n(t+h)=\left( \frac{2C}{h}+G \right) v_n(t)-\frac{C}{h}v_n(t-h)\\
~~~~~~~~~~~~~~~~~~~~~~~~+hL_sv_d-A_i^T\left(I_s(t+h)-I_s(t)\right).
\end{gathered}
\end{equation}
Here $L=A_l^T\mathcal{L}^{-1}A_l$, $L_s=A_l\mathcal{L}^{-1}A_{ls}$. $\mathcal{L}$ is a diagonal matrix with all its diagonal elements being nonzero, hence it is invertible. $\left( \frac{C}{h}+G+hL \right)$ is called system matrix.

\subsection{Positive definiteness}

In the RLC power grid model analysis, some assumptions on the topology are made as summarized below.

\begin{enumerate}
\item All the voltage sources have their negative terminals as ground;
\item Any trivial node is connected to other non-ground nodes by one or more branches and such branches can not be all current sources.
\item For any pair of trivial nodes $i$ and $j$, there exists at least one path from $i$ to $j$ and the path passes through trivial nodes only (it does not pass through source nodes).
\end{enumerate}

The first assumption is straightforward and just the case in power grid models. The second assumption is also true since every node in the power grid model is connected to other non-ground nodes by at least one resistor (or capacitor, or inductor). The third assumption may be not satisfied for some special cases. However, in these cases the power grid can be separated to several sub-circuits which are independent of each other. Consider the case in Fig. (\ref{pic_special}). The circuit can be divided to two independent parts, sub-circuit 1 and sub-circuit 2. In the analysis of each sub-circuit, the third assumption is satisfied.

\begin{figure}[t]
\centering
\includegraphics[width=3.5in]{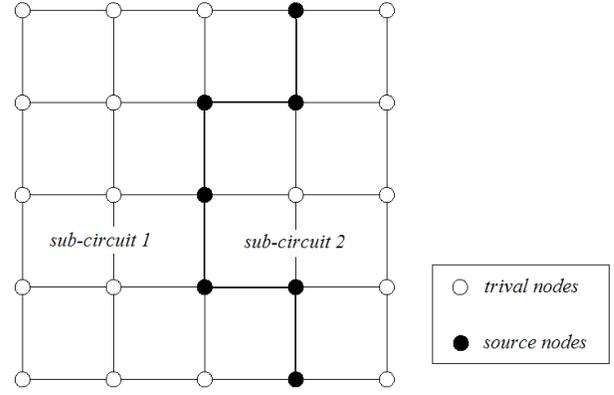}
\caption{A special circuit model which can be divided into two sub-circuits}\label{pic_special}
\end{figure}

\begin{lemma}\label{lem_pos}
Under the topological assumptions for the circuits, the resulting system matrix $M=\left( \frac{C}{h}+G+hL \right)$ is positive definite.
\end{lemma}
\begin{proof}
It is obvious that $M$ is symmetric. And according to the definition of $A_g$, $G(i,i)$ is the sum of conductances connected to trivial node $i$, $\sum\limits_{j\neq i}|G(i,j)|$ is the sum of conductances connected to trivial node $i$ and not connected to source nodes. Hence $|G(i,i)|\geq \sum\limits_{j\neq i}|G(i,j)|$. The same is true for $C$ and $L$, which indicates the system matrix is weakly diagonally dominant with nonnegative diagonal elements.

Besides, according to the topological assumption, any two trivial nodes are connected by a path crossing only trivial nodes. This results in the system matrix to be irreducible. Use the well-known theorem in \cite{young2003iterative}, we conclude that the system matrix $M$ is positive definite.
\end{proof}

\section{Topology-based Transient Analysis of Power Grid}

\subsection{Voltage Update}\label{sec_RLCalg}

\begin{figure}[t]
\centering
\includegraphics[width=2.5in]{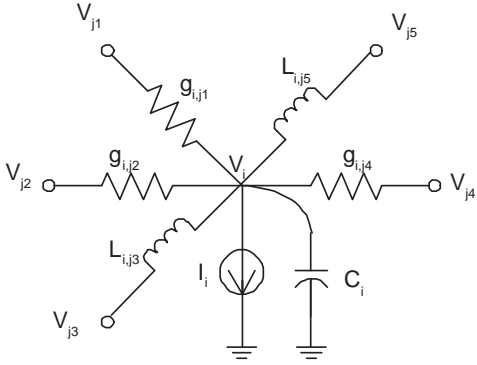}
\caption{A representative node of the RLC power grid}\label{pic_RLCnode}
\end{figure}

Consider a representative node in the RLC power grid model, as shown in Fig. \ref{pic_RLCnode}. Applying Kirchoff's current law at node $i$ we have
\begin{equation}\label{eqn_RLCkir0}
\sum\limits_{j\in N_i^R} g_{ij}(V_i-V_j)+\sum\limits_{j\in N_i^L} I_{ij} +I_i +C_i\frac{dV_i}{dt}=0.
\end{equation}
Here $I_{ij}$ is the current from node $i$ to node $j$, $g_{ij}$ is the conductance between node $i$ and node $j$, $N_i^R$ ($N_i^L$) is the set of indices of nodes connected to node $i$ by a resistor (inductor). Note that $N_i^R\bigcap N_i^L$ may be nonempty. Take derivative on both sides of (\ref{eqn_RLCkir0}) and substitute the branch equations through inductors into it, we obtain
\begin{equation}\label{eqn_RLCkir2}
\sum\limits_{j\in N_i^R} g_{ij}(\frac{dV_i}{dt}-\frac{dV_j}{dt})+\sum\limits_{j\in N_i^L} \frac{1}{L_{ij}} (V_i-V_j) +\frac{dI_i}{dt} +C_i\frac{d^2 V_i}{dt^2}=0.
\end{equation}

\begin{table*}[t]
\scriptsize
\begin{equation}\label{eqn_RLCnode}
\left( \sum\limits_{j\in N_i^R} g_{ij}+h\sum\limits_{j\in N_i^L} \frac{1}{L_{ij}}+\frac{C_i}{h} \right) V_i(t+h)
= \sum\limits_{j\in N_i^R} g_{ij} V_j(t+h) + \sum\limits_{j\in N_i^L} \frac{h}{L_{ij}} V_j(t+h)
+\left(\frac{2C_i}{h}+\sum\limits_{j\in N_i^R}g_{ij} \right)V_i(t)
-\sum\limits_{j\in N_i^R} g_{ij}V_j(t)-I_i(t+h)+I_i(t)-\frac{C_i}{h}V_i(t-h).
\end{equation}
\normalsize
\end{table*}

Here $L_{ij}$ denotes the inductance between node $i$ and node $j$. Using backward Euler method, (\ref{eqn_RLCkir2}) can be discretized  to (\ref{eqn_RLCnode}), as shown below (??in the next page??). Write (\ref{eqn_RLCnode}) more compactly as
\begin{equation}\label{eqn_RLCnodecomp}
V_i(t+h)=\sum\limits_{j\in N_i} w_{ij}V_j(t+h) +K_i(t+h,t,t-h),
\end{equation}
where $w_{ij}$ represents the coefficient of $V_j(t+h)$ and $K_i(t+h,t,t-h)$ is a known value corresponding to currents at $t+h, t$ and nodal voltages $t,t-h$. Based on (\ref{eqn_RLCnodecomp}), we can use an iterative method to solve the power grid model. In each iteration, update the nodal voltage at node $i$ by
\begin{equation}\label{eqn_RLCupdate}
\begin{gathered}
V_i^{(k)}(t+h)=\sum\limits_{j\in L_i}w_{ij}V_j^{(k)}(t+h)+\sum\limits_{j\in U_i}w_{ij}V_j^{(k-1)}(t+h)\\+K_i(t+h,t,t-h).
\end{gathered}
\end{equation}
Here $L_i (U_i)$ is defined as the set of indices that satisfy (i) the corresponding nodes are connected to node $i$ and (ii) they are smaller (larger) than $i$. $k$ denotes the iteration number.

\subsection{Initial conditions}
As it can be seen from (\ref{eqn_systemeqn}) and (\ref{eqn_RLCupdate}), the nodal voltages at $t+h$ depend on the circuit states at both $t$ and $t-h$, which is the requirement of the ``second-order'' characteristic of the circuit. If the initial conditions $v_c(0)$ and $i_l(0)$ are given, we can calculate $v_n(t)$ at any time after $0$ by iteratively updating nodal voltages using (\ref{eqn_RLCupdate}). Given the initial conditions $v_n(0)$ and $i_l(0)$, $v_n(0)$ and $v_n(h)$ can be calculated in three steps.

\begin{enumerate}
\item Treat capacitors as voltage sources whose voltages are $v_c(0)$, inductors as current sources whose currents are $i_l(0)$. Perform iterative DC analysis as \cite{zhong2005fast} at time $0$, obtain $v_n(0)$, $i_c(0)$ and $v_l(0)$;
\item Calculate $v_c(h)$ and $i_l(h)$ according to $\mathcal{C}\frac{v_c(h)-v_c(0)}{h}=i_c(0)$, $\mathcal{L}\frac{i_l(h)-i_l(0)}{h}=v_l(0)$;
\item Perform iterative DC analysis at time point $h$ and obtain $v_n(h)$.
\end{enumerate}

After the three steps we obtain the voltages of trivial nodes ($v_n(0)$ and $v_n(h)$), which are the basis of the topology-based transient analysis algorithm in the next subsection.

\subsection{Algorithm Description}
The topology-based algorithm is proposed as Algorithm \ref{alg_RLC}. No matrix construction or manipulation is required in the algorithm, which dramatically reduces the storage memory and computation time. In the algorithm, DC analysis is performed first to obtain the initial conditions. Then, the nodal voltages at each time step is solved iteratively, based on the information at the previous two time steps. The vectors $V(s)$ and $I(s)$ represent the nodal voltages and current sources at time $s\times h$. The convergence and computational complexity of the proposed algorithm are discussed in Section \ref{sec_RLCcon}.

\begin{algorithm}\label{alg_RLC}
\caption{Topology-based transient analysis of power grid}
\KwIn{Initial capacitor voltages $v_c(0)$, inductor currents $i_l(0)$, step length $h$, maximum step $s_{total}$, currents sources $I(s)$ ($s=0,1,...,s_{total}$), error tolerance $tol$;}
\KwOut{Nodal voltages $V(s)$ ($s=0,1,...,s_{total}$);}

\begin{algorithmic}[1]
\STATE Calculate $V(0)$ and $V(1)$ using iterative DC analysis;
\FOR {$s=2$,$...$,$s_{total}$}
\STATE $V^{(0)}(s)=V(s-1)$;
\STATE $k=0$;
\REPEAT
\STATE $k=k+1;$
\STATE Update $V_i^{(k)}(s)$ using equation (\ref{eqn_RLCupdate});
\UNTIL {$\|V^{(k)}(s)-V^{(k-1)}(s)\|<tol$}
\STATE $V(s)=V^{(k)}(s);$
\ENDFOR

\end{algorithmic}

\end{algorithm}

\subsection{Convergence}\label{sec_RLCcon}

\begin{theorem}\label{thm_conv}
The solution $V(s)$ ($s=0,1,...,s_{total}$) of Algorithm \ref{alg_RLC} converges to the accurate nodal voltages $v_n(s\times h)$.
\end{theorem}

\begin{proof}
Compare (\ref{eqn_RLCnodecomp}) and (\ref{eqn_RLCupdate}) with (\ref{eqn_systemeqn}), we can see that Algorithm \ref{alg_RLC}
is equivalent to the Gauss-Seidel iterative solution of the matrix equation (\ref{eqn_systemeqn}). Because the system matrix is symmetric positive definite, the Gauss-Seidel iteration converges (refer to Theorem 10.2.1 of \cite{golub1996matrix}). Therefore Algorithm \ref{alg_RLC} is guaranteed to converge.
\end{proof}

We can employ the SOR (successive over-relaxation) method to accelerate the convergence of Algorithm \ref{alg_RLC}. Using the SOR-like method, the nodal voltage updating formula (\ref{eqn_RLCupdate}) is adapted to
\begin{equation}\label{eqn_SOR}
\begin{gathered}
V_i^{(k)}(t+h)=\omega\tilde{V}_i^{(k)}(t+h)+(1-\omega) V_i^{(k-1)}(t+h).
\end{gathered}
\end{equation}
Here $\tilde{V}_i^{(k)}(t+h)$ is the updated voltage calculated through (\ref{eqn_RLCupdate}). $\omega$ is called relaxation parameter. Using (\ref{eqn_SOR}) to update the nodal voltage, we obtain a new algorithm (named as Algorithm 2). If $0<\omega<2$, Algorithm 2 is guaranteed to converge. With appropriately chosen $\omega$, the convergence procedure can be accelerated. We refer the readers to \cite{young1970convergence}, \cite{young1972generalizations} for the optimal choice of $\omega$.

\section{Conclusion}
A topology-base power grid transient analysis algorithm has been proposed. Nodal analysis has been adopted to analyze the topology which is mathematically equivalent to iteratively solving a positive semi-definite linear equation. The convergence of the method has been proved.


\end{document}